\documentclass[titlepage, 12pt]{article}

\usepackage{authblk}
\usepackage[margin=1in]{geometry}

\usepackage{algorithm}
\usepackage{algpseudocode}
\usepackage{amsbsy, amsfonts,
amsmath, amssymb, amsthm}
\newtheorem{prop}{Proposition}
\usepackage{multirow}
\usepackage{xcolor}

\usepackage{dsfont}

\usepackage{bm}

\usepackage{booktabs, threeparttable}
\usepackage{setspace}

\usepackage{natbib}

\usepackage{color}

\usepackage{graphicx}
\usepackage{physics}

\usepackage{tikz}
\usetikzlibrary{matrix}

\usepackage[colorlinks=true, citecolor=blue]{hyperref}
\usepackage{mathrsfs}

\usepackage[pagewise]{lineno}
\newcommand*\patchAmsMathEnvironmentForLineno[1]{%
  \expandafter\let\csname old#1\expandafter\endcsname\csname #1\endcsname
  \expandafter\let\csname oldend#1\expandafter\endcsname\csname end#1\endcsname
  \renewenvironment{#1}%
  {\linenomath\csname old#1\endcsname}%
  {\csname oldend#1\endcsname\endlinenomath}}%
\newcommand*\patchBothAmsMathEnvironmentsForLineno[1]{%
  \patchAmsMathEnvironmentForLineno{#1}%
  \patchAmsMathEnvironmentForLineno{#1*}}%
\AtBeginDocument{%
  \patchBothAmsMathEnvironmentsForLineno{equation}%
  \patchBothAmsMathEnvironmentsForLineno{align}%
  \patchBothAmsMathEnvironmentsForLineno{flalign}%
  \patchBothAmsMathEnvironmentsForLineno{alignat}%
  \patchBothAmsMathEnvironmentsForLineno{gather}%
  \patchBothAmsMathEnvironmentsForLineno{multline}%
}

\allowdisplaybreaks 

\newcommand{\ind}{d^{\,(\rm in)}}
\newcommand{\outd}{d^{\,(\rm out)}}
\newcommand{\ins}{s^{(\rm in)}}
\newcommand{\outs}{s^{(\rm out)}}
\newcommand{\tildeins}{\tilde{s}^{(\rm in)}}
\newcommand{\tildeouts}{\tilde{s}^{(\rm out)}}
\newcommand{\adj}{\bm{A}}
\newcommand{\wei}{\bm{W}}
\newcommand{\ER}{Erd\"{o}s--R\'{e}nyi}
\newcommand{\BA}{Barab\'{a}si--Albert}

\newcommand{\norminw}{\tilde{w}^{(\rm in)}}
\newcommand{\normoutw}{\tilde{w}^{(\rm out)}}

\newcommand{\pkg}[1]{{\normalfont\fontseries{b}\selectfont #1}}
\let\proglang=\textsf

\graphicspath{{./}{../figures/}{../figures/PAwd_loc/}{../figures/PAwd/}}

\begin{document}

\title{Assortativity measures for weighted and directed networks}

\author[1]{Yelie Yuan}
\affil[1]{Department of Statistics, University of Connecticut,
Storrs, CT 06269}
\author[1]{Jun Yan}
\author[2,$\ast$]{Panpan Zhang}
\affil[2]{Department of Biostatistics, Epidemiology and
  Informatics, University of Pennsylvania, Philadelphia, PA 19104}
\affil[$\ast$]{Corresponding author: 
\href{mailto:panpan.zhang@pennmedicine.upenn.edu}{panpan.zhang@pennmedicine.upenn.edu}}

\maketitle

\doublespacing

\begin{abstract}
  Assortativity measures the tendency
  of a vertex in a network being connected by other vertexes with
  respect to some vertex-specific features.  Classical assortativity
  coefficients are defined for unweighted and undirected networks with respect
  to vertex degree. We
  propose a class of assortativity coefficients that capture the assortative
  characteristics and structure of weighted and directed networks
  more precisely. The vertex-to-vertex strength correlation is used
  as an example, but the proposed measure can be applied to any pair of
  vertex-specific features. The effectiveness of the proposed measure is
  assessed through extensive simulations based on prevalent random
  network models in comparison with existing assortativity measures.
  In application World Input-Ouput Networks,
  the new measures reveal interesting insights that would
  not be obtained by using existing ones. An implementation is
  publicly available in a \proglang{R} package \pkg{wdnet}.
\end{abstract}

\section{Introduction}
\label{Sec:intro}

In traditional network analysis, {\em assortativity} or {\em
assortative mixing} \citep{Newman2002assortative} is a measure
assessing the preference of a vertex being connected (by edges) with
other vertexes in a network. The measure reflects
the principle of {\em homophily} \citep{McPherson2001birds}---the
tendency of the entities to be associated with similar partners in a
social network. The primitive assortativity measure proposed by
\citet{Newman2002assortative} was defined to
study the tendency of connections between nodes based on their {\em
degrees}, which is why it is also called degree-degree correlation
\citep{VanderHofstad2014degree}. Assortativity can
be defined on any vertex feature, although vertex degree is
the most popular \citep{Newman2002assortative, Newman2003mixing}.
The measure has many applications in different types of
networks; for example, the collaborative
network among the authors on the arXiv preprint platform
\citep{Catanzaro2004assortative}, the transaction network of
country-sectors in the global economy \citep{Cerina2015world}, the
bilateral trade relationships among the countries in the
International Trade Network \citep{Abbate2018distance}, and the
associations among the geographical regions in an earthquake network
\citep{Abe2006complex}.

Classical assortativity measures and their extensions are defined for
unweighted, undirected networks. \citet{Piraveenan2008local}
introduced a vertex-based local assortativity measure quantifying
the contribution to the global assortativity of the network
by each vertex. \citet{VanMieghem2010influence}
reformulated the measure of \citet{Newman2002assortative},
which led to a degree-preserving
rewiring algorithm for unweighted and undirected networks.
Mathematically, Newman's assortativity measure is
analogous to the sample {\em Pearson correlation coefficient} in
statistics. In the context of
networks, this measure
has a limitation that its magnitude depends on the
network size \citep{Dorogovtsev2010zero, Raschke2010measuring}.
\citet{Litvak2013uncovering} gave examples of generative
networks whose assotativity coefficients are expected to be negative
but Newman's measure gives zero when the network
sizes are large. This motivated \citet{Litvak2013uncovering} to
propose a rank-based assortativity measure analogous to the sample
\emph{Spearman correlation coefficient}, which partially
overcomes the limitation. All the definitions mentioned so far only
involve directly connected vertices. To account for the possibility
of a strong association induced by a multi-step path between
two vertices, \citet{Arcagni2017higher} introduced a series of
high-order assortativity measures quantifying the association among
the vertices through a few network properties, including paths,
shortest paths, and random walks with fixed length.

When a network is weighted (i.e., non-unit weights are assigned to
the edges), it is necessary to extend Newman's assortativity
appropriately to a weighted version. Ignoring the weights may lead to
misspecification of the network model and biased inferences.
\citet{Leung2007weighted} introduced a weighted assortativity
coefficient for weighted networks, where each edge weight is
interpreted as the intensity of the interaction between the ending
vertexes.
Specifically, they used the world-wide airport network and the
scientist collaboration network \citep{Barrat2004thearchitecture} as
examples. More recently, \citet{Arcagni2019extending} extended
the concept of high-order assortativities \citep{Arcagni2017higher}
to a class of assortativity measures applicable to weighted
networks by building a connection between
assortativity and {\em centrality}~\citep{Newman2010networks} via
edge weights. The new
class
includes four cases, where cases~1 and~3 were respectively
generalized from the counterparts given in
\citet{Newman2002assortative} and
\citet{Leung2007weighted}, and cases~2 and~4 were further extended
from cases~1 and~3, respectively.

When edge direction needs to be accounted, assortativity is expected
to distinguish the features of source and target vertices. Such
features can be either in-degree or out-degree.
\citet{Newman2003mixing} extended the
classical assortativity \citep{Newman2002assortative} to an out-in
assortativity. In the context of directed biological
networks, \citet{Piraveenan2012assortative}
proposed out-out and in-in assortativity measures, and named them,
respectively, out-assortativity and in-assortativity.
The out-assortativity and in-assortativity of the
studied directed biological networks generally higher
than out-in assortativity coefficient.
\citet{Foster2010edge} introduced four kinds of degree-degree
correlations, $\left(\rm out, in\right)$, $\left(\rm out,
out\right)$, $\left(\rm in, out\right)$ and $\left(\rm in,
in\right)$, where the $\left(\rm out, in\right)$ correlation is
equivalent to that in \citet{Newman2003mixing}, and the
$\left(\rm out, out\right)$ and $\left(\rm in, in\right)$
correlations are respectively equivalent to the out- and
in-assortativity in \citet{Piraveenan2012assortative}.
\citet{Foster2010edge} showed that some real networks (e.g., food
webs) displayed assortative and disassortative mixing simultaneously.

Many networks in practice are both directed and weighted. For
instance, consider the \emph{World Input-Output Network} (WION)
defined by a {\em World Input-Output Table}
\citep[WIOT,][]{Timmer2015anillustrated}. Each vertex in a WION represents a
region-sector, i.e., a sector in a country, region, or territory.
A directed edge represents the monetary amount of the transactions
from its source region-sector to its target region-sector, with the amount
defining the edge weight. Despite the applications of network modeling
tools to WIONs \citep[e.g.,][]{Cerina2015world}, there is rather
limited work on assortativity measures for weighted and directed
networks in the literature. The classical assortativity
measure cannot appropriately quantify the assortativity in a WION.
The structure of a WION provides a vehicle to study any feature
between the
source region-sectors and target region-sectors, which could provide
insights into the assortative or disassortative mixing beyond what
has been typically studied.

We summarize the three major contributions of the present paper.
First, we propose a class of assortativity measures for weighted and
directed networks. The proposed measures can be used to analyze the
correlation structure of both network-based properties and
vertex-specific features; hence, they are general. Second, the
necessity of incorporating weight and direction is illustrated via
extensive numerical studies on widespread network models. We assert
that edge weight is sine qua non of the assortativity of weighted
networks, while ignoring it in the computation would fail
to capture network features and possibly lead to incorrect
inference. Finally, we apply the
proposed assortativity measures to the annual series of WIONs from 
2000 to 2004 to investigate the
association among various region-sectors.

\section{Weighted and Directed Assortativity Coefficients}
\label{Sec:novel}

Let $G(V, E)$ be a network consisting of a set of vertices $V$ and a
set of edges
$E$. The cardinalities of $V$ and~$E$, respectively denoted by $|V|$
and $|E|$, are called the {\em order} and the {\em size} of $G$. For
any $i \neq j \in V$, let $e_{ij}$ denote a {\em directed} edge
(excluding self-loop) from vertex~$i$ to $j$. We use
$\adj^{(G)} := (a_{ij})_{|V| \times
|V|}$ to represent the {\em adjacency matrix} of an unweighted network $G(V, E)$
such that $a_{ij} = 1$ if
$e_{ij} \in E$ and $a_{ij} = 0$, otherwise. When $G(V, E)$ is weighted,
we use $\wei^{(G)} := (w_{ij})_{|V| \times |V|}$ to represent the
weighted analogue of $\adj^{(G)}$. That is, $w_{ij} > 0$ is the weight
of $e_{ij}$ if $e_{ij} \in E$; $w_{ij} = 0$, otherwise. Apparently,
$\adj^{(G)}$ and $\wei^{(G)}$ are
identical for unweighted networks. For ease of notation, we
suppress the superscripts of
$\adj^{(G)}$ and $\wei^{(G)}$ in the sequel.

The {\em degree} of a vertex in an undirected network is
the number of edges incident to it. In directed networks, the {\em
out-degree} and {\em in-degree} of a vertex are the number of edges
emanating from and pointing to it, respectively. We use
$d_i$, $\ind_i$ and $\outd_i$ for the
degree, the in-degree, and the out-degree of vertex~$i$,
respectively. When
edge weights are present, we define {\em out-strength} and {\em
in-strength} of a vertex as the total weights of the edges emanating
from and pointing to it. The total {\em strength} (or simply
strength) of a vertex is the sum of its in-strength and out-strength.
The strength, in-strength, and out-strength of vertex $i$ are
respectively denoted by $s_i$, $\ins_i$ and $\outs_i$.

In classical network analysis, assortativity is usually referred to
a degree-degree correlation. In general, assortativity can be used
as a tool measuring the association between any pair of vertex
features. Let $X$ and $Y$ be two quantitative features
for all the vertices in a weighted and directed network
$G(V, E)$. Let $(X_i, Y_i)$ be the two features for each vertex $i \in V$.
Our weighted and directed assortativity measure based on the sample
Pearson
correlation coefficient is defined as
\begin{equation}
	\label{eq:cor}
	\rho_{X,Y}(G) =
        \frac{\sum_{i,j \in V} w_{ij}
          (X_i - \bar{X}_{\rm sou}) (Y_j - \bar{Y}_{\rm tar})}
        {W \sigma_{X}	\sigma_{Y}},
\end{equation}
where $W := \sum_{i,j \in V} w_{ij}$ is the sum of edge weights,
$\bar{X}_{\rm sou}$ is the average of feature $X$ based on source
vertices, $\bar{X}_{\rm sou} = \sum_{i,k \in V} w_{ik} X_i / W$ is
the average of feature $X$ based on source vertices,
$\bar{Y}_{\rm tar} = \sum_{k,j \in V} w_{kj} Y_j /W$
is the average of feature $Y$ based on target vertices,
$\sigma_X = \sqrt{\sum_{i,k \in V} w_{ik} (X_i - \bar{X}_{\rm sou})^2 / W}$
and
$\sigma_Y = \sqrt{\sum_{k,j \in V} w_{kj} (Y_j - \bar{Y}_{\rm tar})^2 / W}$
are the associated standard deviations, respectively.
The two features $X$ and $Y$
could be the same feature such as the total strength.

Specifically, we look into strength-strength correlations by
incorporating edge weights with the four (directed but unweighted)
assortativity measures proposed by \citet{Foster2010edge}. For each
edge, let $\alpha, \beta \in \left\{{\rm in}, {\rm out} \right\}$
index the strength type of the vertices at the two ends. We
accordingly have four kinds of extended assortativity measures for
a weighted and directed network $G$, denoted
$\rho_{\alpha,\beta}(G)$. Note that
Equation~\eqref{eq:cor} is well defined for any $w_{ij} \in \mathbb{R}^{+}$
for $i,j \in V$. In the special case where all edge weights $w_{ij}$'s are
positive integers, each edge can be regarded as a {\em multi-edge},
structurally equivalent to $w$ edges with unit weight.

According to Equation~\eqref{eq:cor}, we define the
$(\alpha, \beta)$-type assortativity coefficient of a weighted and
directed network $G(V,E)$ as follows:
\begin{equation}
\label{eq:fourd}
\displaystyle
\rho_{\alpha, \beta}(G) =
\frac{
  \sum_{i, j \in V} w_{ij}
  \left[\left(s_i^{(\alpha)} - \bar{s}_{\rm sou}^{\, (\alpha)} \right)
    \left(s_j^{(\beta)} - \bar{s}_{\rm tar}^{\, (\beta)} \right)\right]
}
{W \sigma_{\rm sou}^{(\alpha)} \sigma_{\rm tar}^{(\beta)}},
\end{equation}
where $s_{i}^{(\alpha)}$ is the $\alpha$-type strength of
the source vertex $i$,
\begin{equation*}
	\bar{s}_{\rm sou}^{\, (\alpha)} = \frac{\sum_{i,j \in V} w_{ij}
		s_i^{\, (\alpha)}}{W} = \frac{\sum_{i \in V} \outs_i
		s_i^{(\alpha)}}{W}
\end{equation*}
is the weighted mean of the
$\alpha$-type strengths of the source vertices and
\begin{equation*}
	\sigma_{\rm sou}^{(\alpha)} = \sqrt{\frac{\sum_{i,k \in V}
	w_{ik} \left(s_i^{(\alpha)} - \bar{s}_{\rm sou}^{\, (\alpha)}
	\right)^2}{W}} = \sqrt{\frac{\sum_{i \in V} \outs_i
	\left(s_i^{(\alpha)} -
		\bar{s}_{\rm sou}^{\, (\alpha)} \right)^2}{W}}
\end{equation*}
is the associated weighted standard deviation. The counterparts
$s_j^{(\beta)}$, $\bar{s}_{\rm tar}^{\, (\beta)}$ and
$\sigma_{\rm tar}^{(\beta)}$ are
defined analogously for the target vertex~$j$, where the weights in
the weighted mean and the weighted standard deviation are replaced
with $\ins_j$. An illustrative example of the proposed assortativity
measures is given in Appendix~\ref{Append:example}.

The value of $\rho_{\alpha, \beta}(G)$ is bounded between $-1$ and
$1$. The sign indicates the direction of the assortativity. A larger magnitude
indicates a higher assortativity between the vertices with large $\alpha$-type
strength and those with $\beta$-type strength in $G$. The boundary
values are attainable in some extreme cases. Several examples are
presented in Appendix~\ref{Append:boundary}. The definition of
$\rho_{\alpha,\beta}(G)$ is a sample measure, which should be
fluctuating around
the population measure of the network model that generates the
observed
network. When $G$ is completely random, for instance, as generated from the \ER\
random network model \citep{Erdos1959on}, the population measure should be
zero as there is no tendency of vertex connection. The coefficient will be
close to zero when the network size is sufficiently large.

The proposed assortativity coefficient in Equation~\eqref{eq:fourd}
is well-defined for
networks that are unweighted, undirected or both. Specifically, it is equivalent
to that given in \citet[case 4]{Arcagni2019extending} for weighted but
undirected network. When the vertex strengths are
replaced with degrees, the proposed assortativity coefficient is
simplified to \citet[case~3]{Arcagni2019extending}, mathematically
identical to the measure introduced by \citet{Leung2007weighted}.
See Proposition~\ref{Prop:weightundirect} in Appendix~\ref{Append:math}
for a mathematical justification. For unweighted but directed networks, we
have $w_{ij} = 1$ for all $e_{ij} \in E$, implying that $W$ counts
the number of edges. It is obvious that Equation~\eqref{eq:fourd} is
equivalent to the four-directed assortativity given by
\citet{Foster2010edge}. For some specific choices of $\alpha$ and
$\beta$, Equation~\eqref{eq:fourd} is equivalent to some results
developed in \citet{Newman2003mixing, Piraveenan2012assortative},
demonstrated in Appendix~\ref{Append:math} as well. For unweighted
and undirected networks, Equation~\eqref{eq:fourd} is further
simplified to the measure proposed by \citet{Newman2002assortative},
shown in Appendix~\ref{Append:math}.

\section{Illustrations}
\label{Sec:sim}

We illustrate the proposed assortativity coefficients in the context of a few
widely used random network models extended to allow direction and weight in
edges. The new assortativity measures are shown to
provide more insights for weighted and directed networks.

\subsection{\ER\ Random Networks}
\label{Sec:ermodel}

We extend the classical \ER\ (ER)  random network model
\citep{Erdos1959on, Gilbert1959random} by incorporating independently
generated edge direction and weights.
For each ordered pair of vertices $(i, j)$, $i\neq j \in V$, let $B_{ij}= 1$
if there is a directed edge from $i$ to $j$ and $B_{ij} = 0$ otherwise.
The $B_{i}$'s are independent {\em Bernoulli} variables with rate $p$.
The weights of the directed edges are independently drawn from a
distribution with support on any subset of $\mathbb{R}^{+}$.
Because of the independently generated edges and weights,
the assortativity of this network is expected to be~$0$ when the network
order~$n$ is large.

The specific settings our  extended ER model are as follows.
The weight distribution was set to be discrete {\em uniform} distribution
with support $\{1, 2, \ldots, \theta\}$ with parameter $\theta \in \mathbb{N}$.
Let $U_{ij}$ be the weight corresponding to edge $i\to j$ if $B_{ij} = 1$ and
zero otherwise. The in-strength and out-strength of vertex $i$ is governed by a
random variable $S_i^{(\text{out})} := \sum_{k \ne i}^{n} (U_{ki}  B_{ki})$ and
$S_i^{(\text{in})} := \sum_{k \ne i}^{n} (U_{ik}  B_{ik})$, respectively,
which share the same distribution with mean $(n- 1) p \theta/2$.
The model parameters were set to be $p = 0.2$ and $\theta = 10$.
We considered networks of order $ n \in \{10, 20, \ldots, 140, 150\}$.
In each configuration, a total of 2,000 weighted and directed ER networks were
generated, and all all four kinds of assortativity coefficients were computed
for each replicate.

\begin{figure}[tbp]
  \centering
  \includegraphics[width=\textwidth]{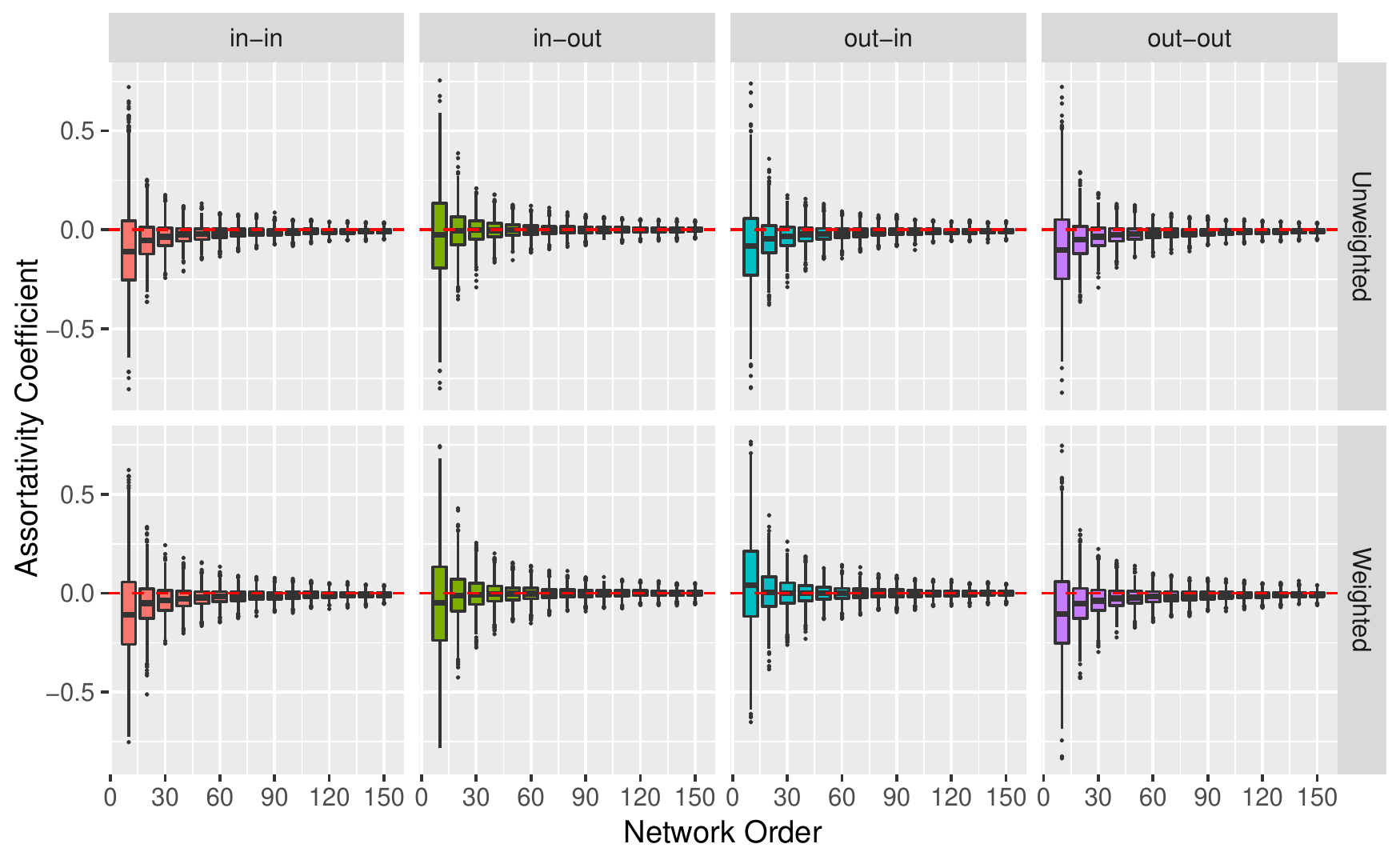}
  \caption{Boxplots of four kinds of weighted and unweighted assortativity
  coefficients for 2,000 weighted and directed ER random networks; the red
  dashed line in each panel is $\rho  = 0$.}
  \label{fig:ER_model}
\end{figure}

Figure~\ref{fig:ER_model} shows the boxplots of the assortativity coefficients
for the 2,000 networks against the network order~$n$. For
comparison, we also included the four kinds of directed but
unweighted assortativity coefficients \citep{Foster2010edge}.
All the assortativity coefficients converge to~$0$
quickly. For a moderate network order $n = 150$, all
the assortativity coefficients are centered around $0$ with small
variation. There is no significant difference between any pair of
weighted and unweighted assortativity measures. No obvious difference
in convergence rate is observed, either.
Since the extended ER random networks are completely
random networks, where the edges are added independently, the fact
of no tendency of vertex connection is well reflected by the assortativity
coefficients.

\subsection{\BA\ Random Networks}
\label{Sec:bamodel}

The Barab{\'a}si--Albert (BA) random network model
\citep{Barabasi1999emergence} defines a generative network
with a {\em preferential attachment} mechanism.
Starting from a seed graph,  at each evolutionary
step, a new vetex is connected with an existing vertex with
probability proportional to its degree. BA networks are scale-free networks
because their degree sequence obeys a {\em power law}
\citep{Bollobas2001thedegree}.
\citet{VanderHofstad2014degree} showed that the
assortativity of an undirected and unweighted BA network is
asymptotically~$0$. The assortativity of weighted and
directed BA networks has not been investigated before.

Our construction of weighted and directed BA networks is extended from
the algorithm of \citet{Wan2017fitting}. The seed
graph consists of two vertices connected by a directed edge. The
edge addition process is controlled by tow parameters
$\alpha, \gamma \in [0, 1]$
subject to $\alpha + \gamma = 1$. With probability~$\alpha$,
an edge is added from a newcomer to a sampled vertex in
the existing network; with probability~$\gamma$, an edge is added
from a sampled vertex in the existing network pointing in the
newcomer. The weight of each edge is drawn independently from a
pre-specified distribution with positive support.
The probability of an existing vertex being sampled at
each step is proportional to its in-strength or
out-strength, instead of its in-degree or out-degree as
in~\citet{Wan2017fitting}; the choice between in-strength and
out-strength depends on the direction of the newly added edge.
Additional tuning parameters $\delta_{\rm out}$ and $\delta_{\rm
in}$ control the growth rate of out-strength and in-strength,
respectively.

We set $\alpha = 0.6$, $\gamma = 0.4$,
and $\delta_{\rm in} = \delta_{\rm out} = 1$,
and drew edge weight from a
discrete uniform distribution over $\{1, 2, \ldots, 10\}$. To investigate
the relation between network order and assortativity, we considered network
order $n \in \{2^4, 2^5, \ldots, 2^{12}\}$. For each setting,
2,000 independent networks were generated.
Figure~\ref{fig:PA_alpha0.6} shows the boxplots of
weighted and unweighted out-in assortativity coefficients;
the plots for out-out, in-in, and in-out
assortativity coefficients are similar and, hence, omitted. Both
weighted and unweighted assortativity
coefficients are negative, with the weighted version closer to zero.
The negative assortativity is expected since new vertices (of out-
or in-strength~$1$) tend to be
connected with the vertices of large out-strength or in-strength
vertices in the existing network. When network size becomes larger, both
assortativity values get closer to zero with small variations.

\begin{figure}[tbp]
	\centering
	\includegraphics[width=\textwidth]{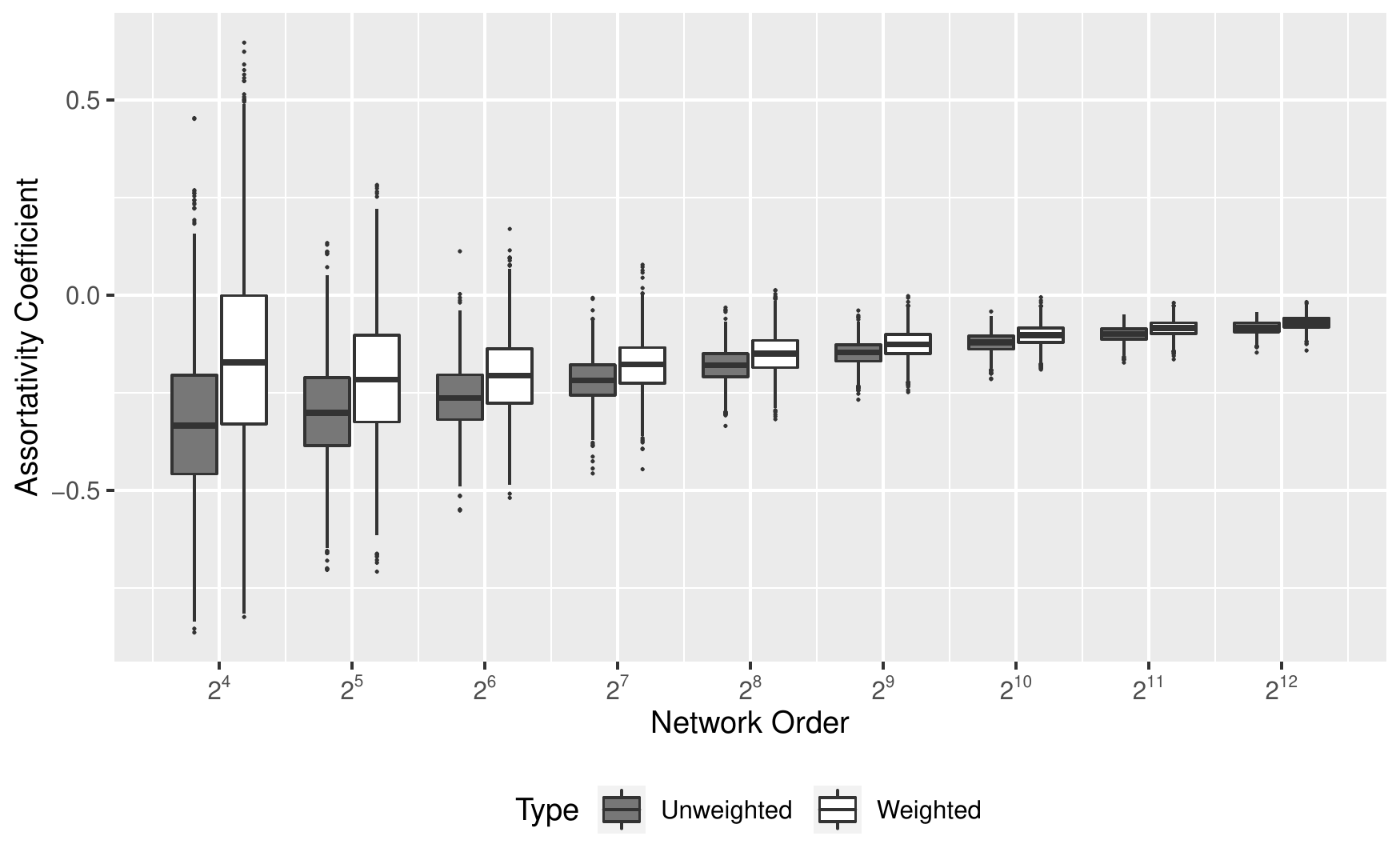}
	\caption{Side-by-side boxplots of weighted and unweighted out-in
	assortativity coefficients for 2,000 BA random networks of
	different size.}
	\label{fig:PA_alpha0.6}
\end{figure}

Edge weight may have significant impact on network assortativity. To
illustrate, we considered adding an edge with a substantially large
weight as network grows. Under the same setting of
$\alpha, \gamma, \delta_{\rm in}$ and $\delta_{\rm out}$, we
generated BA networks with $500$ steps (i.e., equivalently $502$ 
vertices). There is one edge with weight $1,000$, while all of the
others have relatively small weights, drawn uniformly from $\{1, 2,
\ldots, 10\}$. The arrival time of the large-weight edge was set
at
$t \in \left\{10, 20, \ldots, 500 \right\}$.
For each $t$, a total of 2,000 networks were
generated, and the weighted and unweighted out-in assortativity
coefficients were calculated.

\begin{figure}[tbp]
	\centering
	\includegraphics[width=\textwidth]{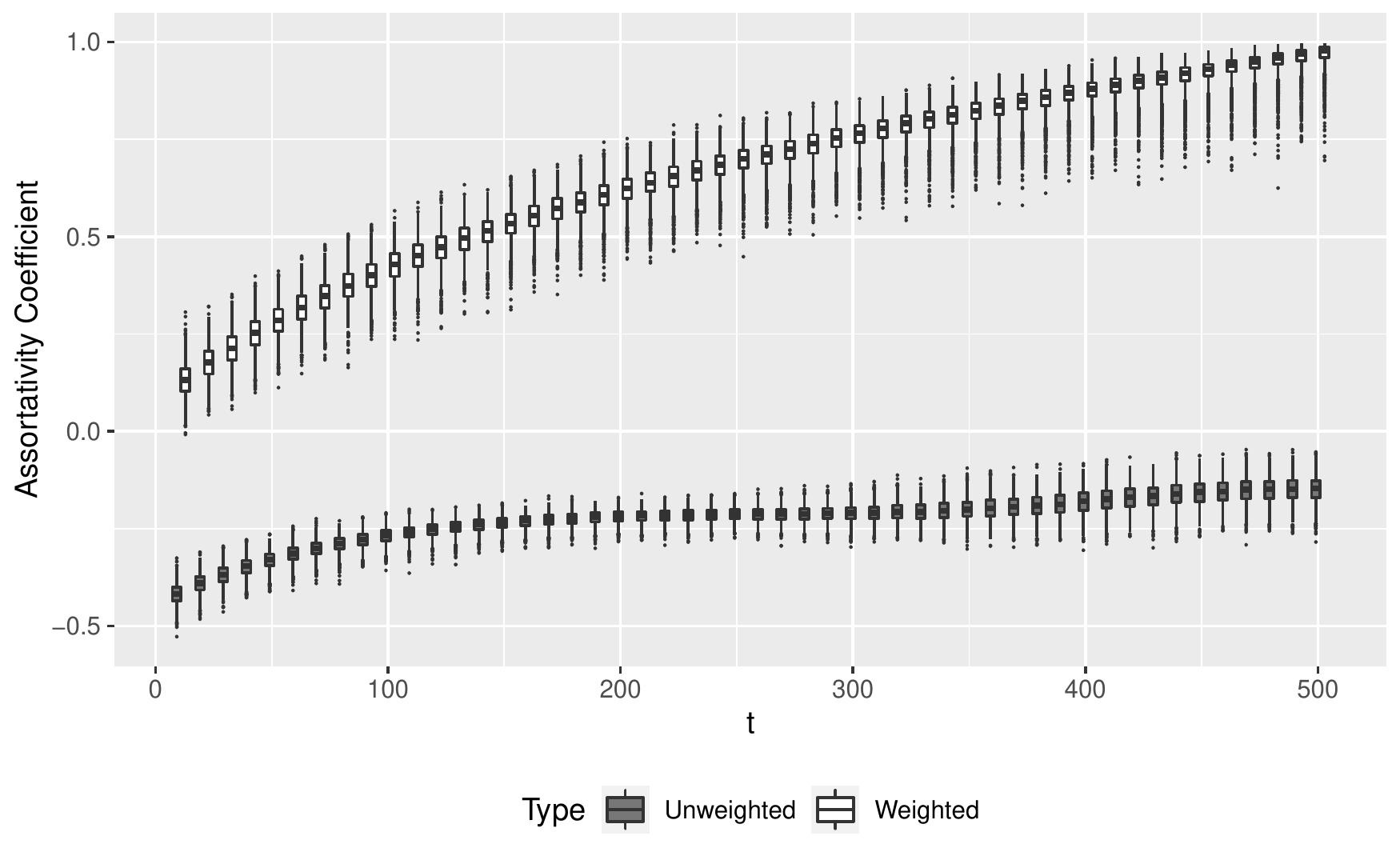}
  \caption{Side-by-side boxplots of weighted and unweighted out-in
  assortativity coefficients of BA random networks (fixed size 500),
  with one large-weight edge joining in at different timestamps.}
	\label{fig:PA_loc_alpha0.6}
\end{figure}

Figure~\ref{fig:PA_loc_alpha0.6} shows the boxplots of the weighted and
unweighted out-in assortativity coefficients against the arrivale time of the
big-weight edge.
Unlike in Figure~\ref{fig:PA_alpha0.6}, here the weighted and unweighted
assortativity coefficients appear to be drastically different.
The weighted assortativity coefficients are positive, but the unweighted
counterparts stay negative. The later the large-weight edge
joins the network, the higher the weighted out-in assortativity
coefficient is. Especially when the large-weight edge is the last one,
it ensures that the edge links a
vertex of large out-strength to a vertex of large in-strength.
Since the weights of all preceding edges are relatively
small, the contributions to the weighted out-in assortativity
is dominated by the last large-weight edge, and thus the
assortativity reaches a high value. When the large-weight edge is
appended to the network at an early time, the vertices at the two
ends at its first appearance have high probabilities to attract the
subsequent newcomers; the weighted assortativity coefficient
remains positive owing to the impact of the large
weight, but its effect is tapered off over time.
The unweighted assortativity coefficients remain negative and show
no response to arrival time of the big-weight edge, which does not
reflect the expected changes.

\subsection{Stochastic Block Models}
\label{Sec:sbm}

The stochastic block model \citep[SBM,][]{Holland1983stochastic}
is a class of network models presenting community structures. The
vertices in a network based on SBM are partitioned into disjoint
communities such that the number of edges among the vertices within
a community is expected to be significantly higher than that among
the vertices from different communities. One of the classical edge
addition algorithms for generating an SBM is a two-step procedure.
Within each community, generate an ER random graph with a relatively
large value of $p$; between communities, generate an edge with
probability $p^{\prime} < p$ between each pair of between-community
vertices. The set of $p$ and $p^{\prime}$ may vary from
community to community.

We considered SBM networks consisting of two
communities of equal size. The link density for each community was
fixed at $p = 0.2$, and the edge weights within each community
were sampled uniformly from $(0, 5)$ and $(5, 10)$, respectively.
The between-community link density was set to be
$p^{\prime} = 0.02$, independent of within-community edges. The
between-community edge weights were set to be identically~$5$.
The edge directions were assigned randomly with equal probability and
independently, for both within- and between-community edges. We
allowed the community size to vary in $\{50, 100, \ldots, 500\}$.
For each setting, we generated 2,000 SBM networks and calculated
their weighted and unweighted out-in assortativity coefficients.

\begin{figure}[tbp]
	\centering
	\includegraphics[width=\textwidth]{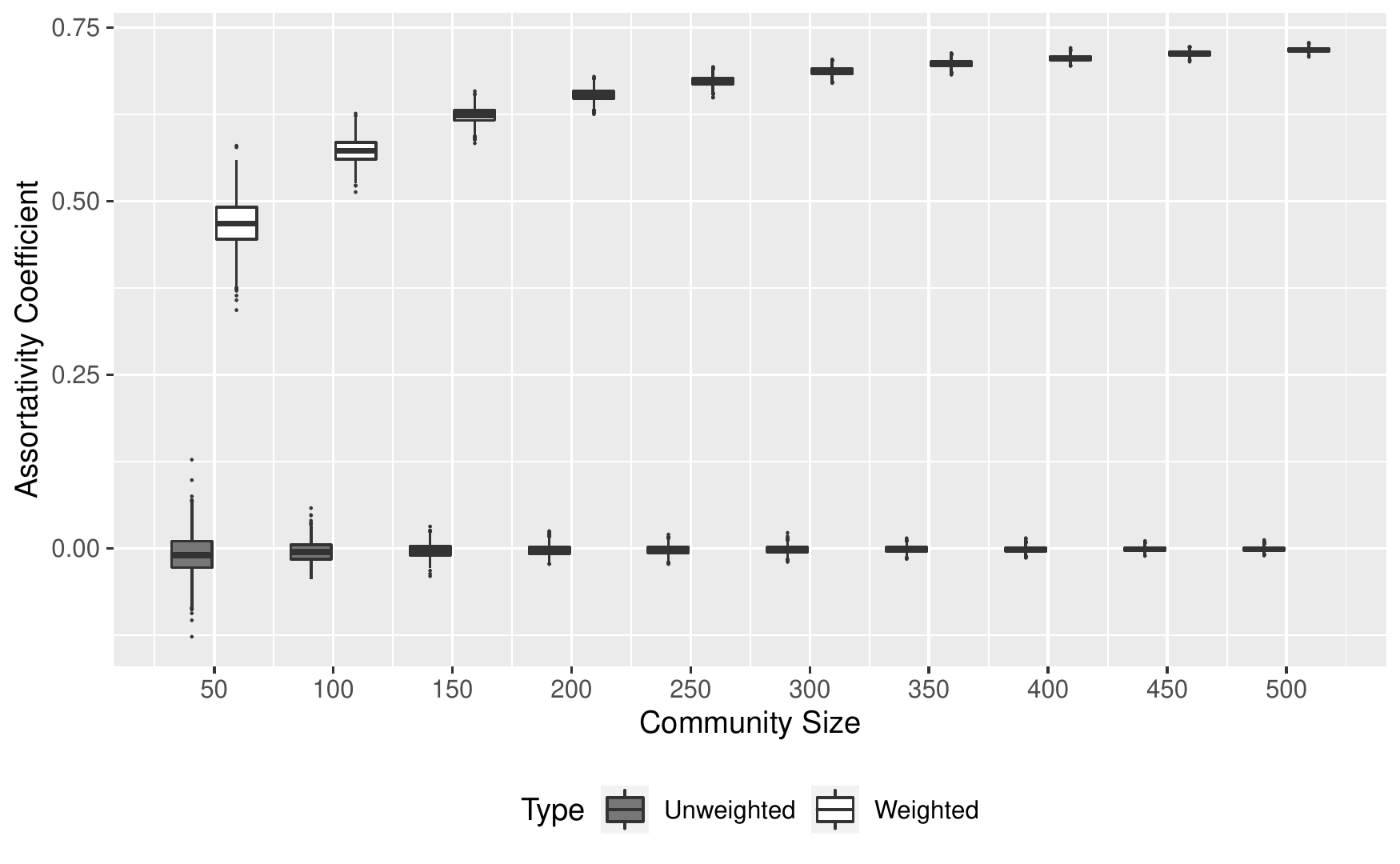}
	\caption{Side-by-side boxplots of weighted and unweighted
	out-in assortativity coefficients based on 2,000 SMB networks
	with different community sizes.}
	\label{fig:SBM1}
\end{figure}

Figure~\ref{fig:SBM1} shows the side-by-side boxplots of the
weighted and unweighted out-in assortativity coefficients
against the community size. The unweighted assortativity
coefficients are centered around~$0$ with
small variations regardless of the change in community size (and
network size). That is expected as the edges are added independently,
and the edge weights are ignored in the computations of unweighted
assortativity coefficients. Each of the simulated networks is
structurally equivalent to a composition of two (mutually
independent) ER random graphs and
an independent Bernoulli model, leading to~$0$ assortativity.
When the edge weights are accounted, however,
the assortativity coefficients are obviously positive with
higher values for the networks of larger size. The
within-community edges
always link the vertices with large out-strength to those with large
in-strength (or the vertices with small out-strength to those with
small in-strength). The between-community link density is relatively
small, and the between-community edges are likely to link two
isolated vertices when the communities are large in size.
Thus, the weighted out-in assortativity
coefficients tend to be positive, especially for large networks.

To investigate the impact of between-community edge
weight and link density on network assortativity, we conducted a
a sensitivity analysis. The changes to the settings are:
the community size was fixed at~$500$;
the between-community link density parameter was set to be
$p' \in \{0.02, 0.03, \dots, 0.1\}$; and
the between-community edge weights were set to be identically
$5k$ with $k \in \{1, 1/2, 1/4, 1/8\}$. For each configuration,
$2,000$ SBMs were generated.  The boxplots of the weighted
out-in assortativity coefficients are summarized in
Figure~\ref{fig:SBM1_2}. As expected, the
assortativity coefficient decreases as $p^{\prime}$ increases for each
given $k$; the decreasing is faster for smaller $k$.
For each given $p^{\prime}$, the assortativity coefficient increases
as $k$ decreases, and the increase is faster for higher $p'$.
The results support that both between-community
edge weight and link density have impact on assortativity values.

\begin{figure}[tbp]
	\centering
	\includegraphics[width=\textwidth]{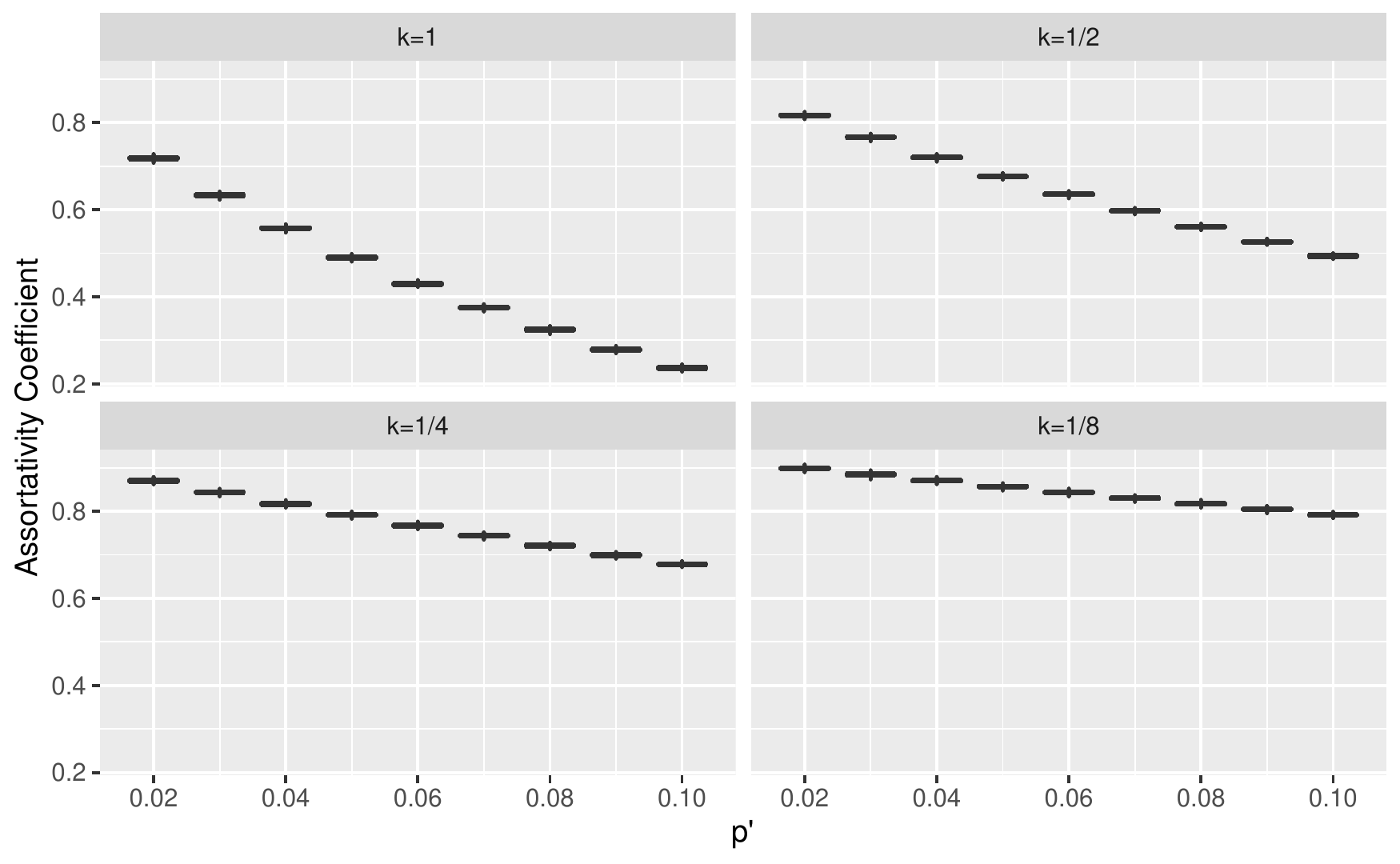}
	\caption{Weighted out-in assortativity
		coefficients of SMB networks by choosing different values of
		ratio $k$ and $p'$.}
	\label{fig:SBM1_2}
\end{figure}

\subsection{Rewired Networks with Given Assortativity}
\label{Sec:targetsim}

A weighted (undirected) network can be rewired so that its assortativity
achieves a predetermined level by extending the algorithm of
\citet{Newman2003mixing}. The details of our rewiring algorithm are 
given in Appendix~\ref{Append:rewire}.
Let $\xi$ be the target assortativity coefficient of a weighted network
of~$n$ vertices. Let $S$ denote a random variable of admissible vertex
strength with support $\mathcal{S}$, where $|\mathcal{S}|< \infty$.
We generated networks with $n=2,000$ vertices whose
strengths are distributed such that
$p_k = \Pr(S = z_k) \propto z_k^{-2.5}\exp{-z_k / 100}$
for $z_k \in \mathcal{S} = \left\{ 10, 11, \dots, 99, 100\right\}$.
The number of rewiring steps was fixed at $5000 \times n$. For each
$\xi \in \left\{0.1, 0.2, 0.3, \dots, 0.8, 0.9\right\}$, 100 networks were
simulated.

\begin{figure}[tbp]
	\centering
	\includegraphics[width=\textwidth]{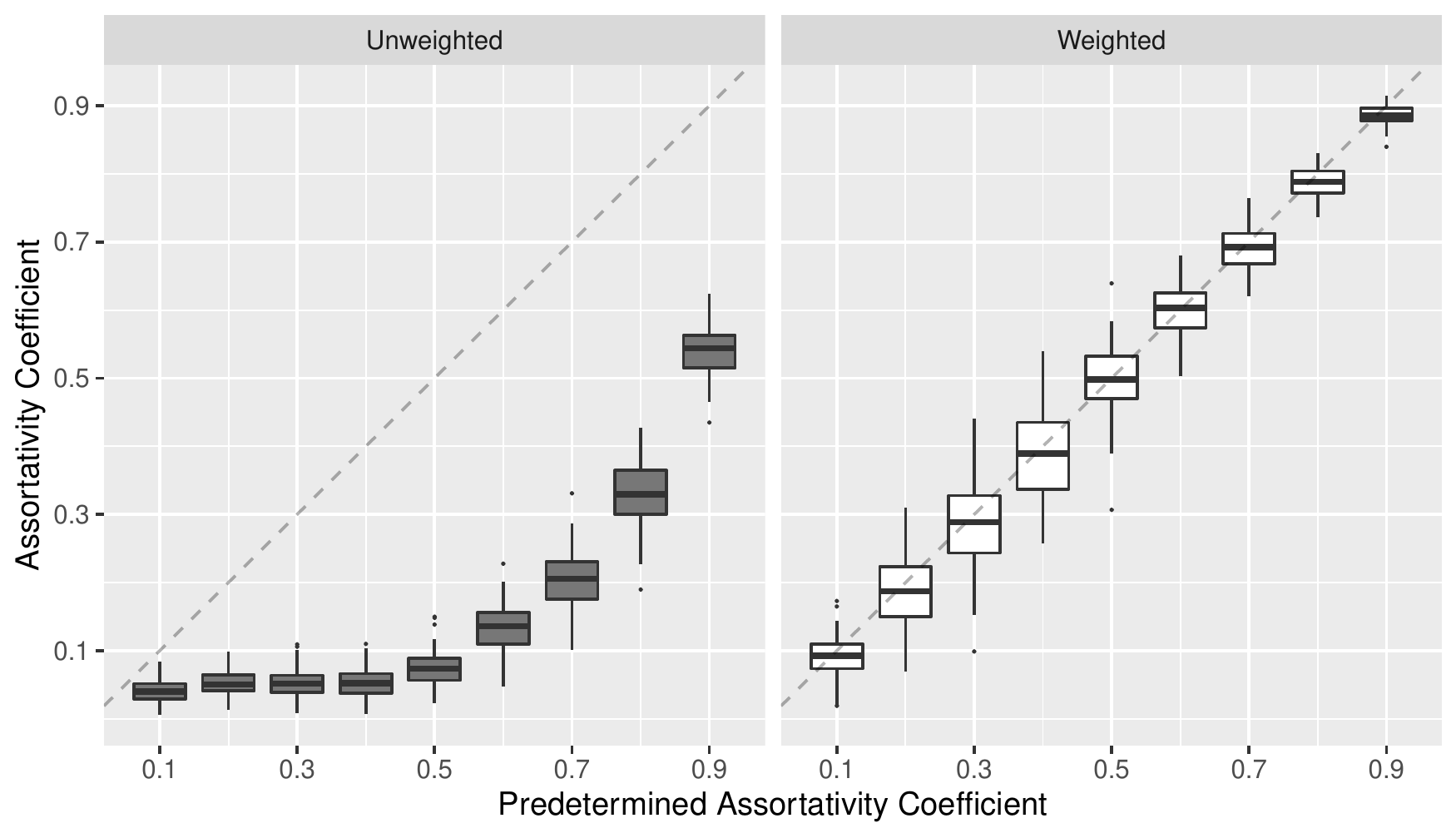}
	\caption{Boxplots of the weighted and unweighted assortativity
	coefficients of	the simulated networks with respect to
	different predetermined values.}
	\label{fig:given_rho2}
\end{figure}

Figure~\ref{fig:given_rho2} shows the boxplots of both weighted and
unweighted assortativity coefficients of the 100 simulated networks
against~$\xi$. The weighted assortativity coefficients (right panel) are
approximately centered on the $45$-degree line, suggesting that the
average of the weighted assortativity coefficients are consistent
with the predetermined values. The boxplots of the
unweighted assortativity coefficients (left panel), however, stay far
away from the $45$-degree line, implying that the unweighted
assortativity coefficient fails to capture the assortative feature of the
weighted network.

\section{Application to WIONs}
\label{Sec:applications}

We apply the proposed measures to WIONs constructed with the 
2016 release of the World Input-Output database
\citep{Timmer2015anillustrated}. There have been a few studies of the
WIONs \citep{Cerina2015world, delRio2017trends, 
Piccardi2017random} with data based on the 2013 release.
The 2016 release is the most recent, covering the period of 2001--2014. For
each year, the WIOT
records the directed economic transactions among 56~sectors of 44~countries,
regions, or territories, the last one of which is the rest of the world (RoW).
Therefore, the resulting WIONs have order 2,464. The
edge weights are in the unit of million US dollars (USD). For the 
purpose of temporal comparison, an adjustment for inflation has been 
applied to the data using the GDP deflators from the World Bank 
(\url{https://data.worldbank.org/indicator/NY.GDP.DEFL.ZS}).
The WIONs are extremely dense (e.g., the link density of the WION of 2014 is
$0.83$), but contain a large amount of edges with
small weights (e.g., the $90$-th percentile of the edge weights in
the WION of 2014 is $0.81$ millions USD). The large amount of edges with small
weights tend to blur the fundamental structure of the network that are of primary
interest. One way to extract the fundamental structure is to use the
{\em backbone} of the WION after discarding noisy edges up to a certain level.
We adopt the filtering procedure introduced by~\citet{Xu2019input}; see details
in Appendix~\ref{Append:backbone}.

\begin{figure}[tbp]
	\centering
	\includegraphics[width=\textwidth]{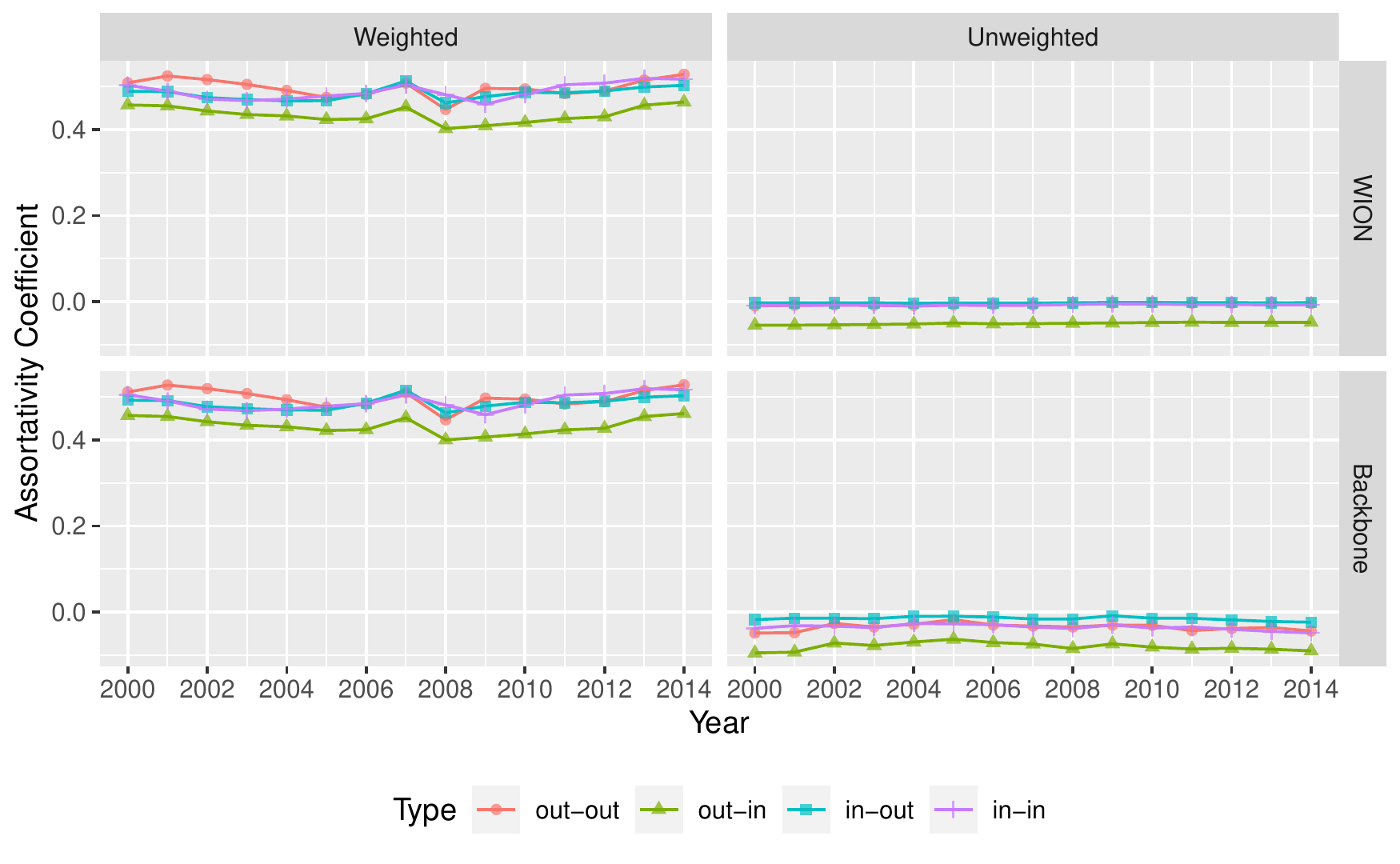}
  \caption{Weighted and unweighted assortativity 
  coefficients of WIONs and their corresponding backbones at level 
  $0.05$ from 2000 to 2014.}
	\label{fig:backbone}
\end{figure}

Figure~\ref{fig:backbone} presents the four kinds of weighted assortativity
coefficients from the whole WIONs and their backbones of level  $0.05$ from
2000 to 2014. For comparison, with unweighted assortativity
coefficients~\citep{Barrat2004thearchitecture} that were used
by~\citet{Cerina2015world} are also included. In the upper-left panel, all
four weighted assortativity measures for the whole WIONs are positive in the
range from $0.40$ to $0.53$, suggesting that the vertices of similar strength
levels are more likely to be connected. For
instance, high-instrength region-sectors like ``construction'' 
usually 
take large inputs from high-outstrength region-sectors like 
manufactures of ``mineral products'', ``basic metals'' and ``wood 
and cork''. Another example is given by ``manufactures of
computer, electronic and optical products'' supplying large outputs 
to some relevant region-sectors such as ``computer programming, 
consultancy and related activities''.
WIONs have revealed a geographical feature that large amount of 
monetary transactions usually occur among the region-sectors in the 
same country, while fewer or substantially smaller transactions are 
observed among the region-sectors across different countries. This 
geographical feature helps explain the positive values of the 
weighted assortativity measures for WIONs. The temporal changes in 
all 
assortativity coefficients 
appear to be quite
similar. Each curve presents a notable from 2007 to 2008 when the global
financial crisis occurred. After 2009, a consistently upward
trend emerges in each curve, suggesting recovery of the global economy.
The counterparts for the backbones show almost the same magnitude and patterns.
The backbones of level 0.05 had only about 5\% of the edges preserved but
account for over
96\% of the total weights. This shows the robustness of the proposed
assortativity coefficients in capturing the fundamental structure of WIONs.
In contrast, the unweighted assortativity coefficients are negative (close to
$0$ in magnitude) with little temporal changes over the years, similar to
those reported by~\citet{Cerina2015world}. The dramatic difference between
the weighted and unweighted results suggests that ignoring link weights in
any kind of assortativity measure may lead to misleading conclusions.

\begin{figure}[tbp]
	\centering
	\includegraphics[width=\textwidth]{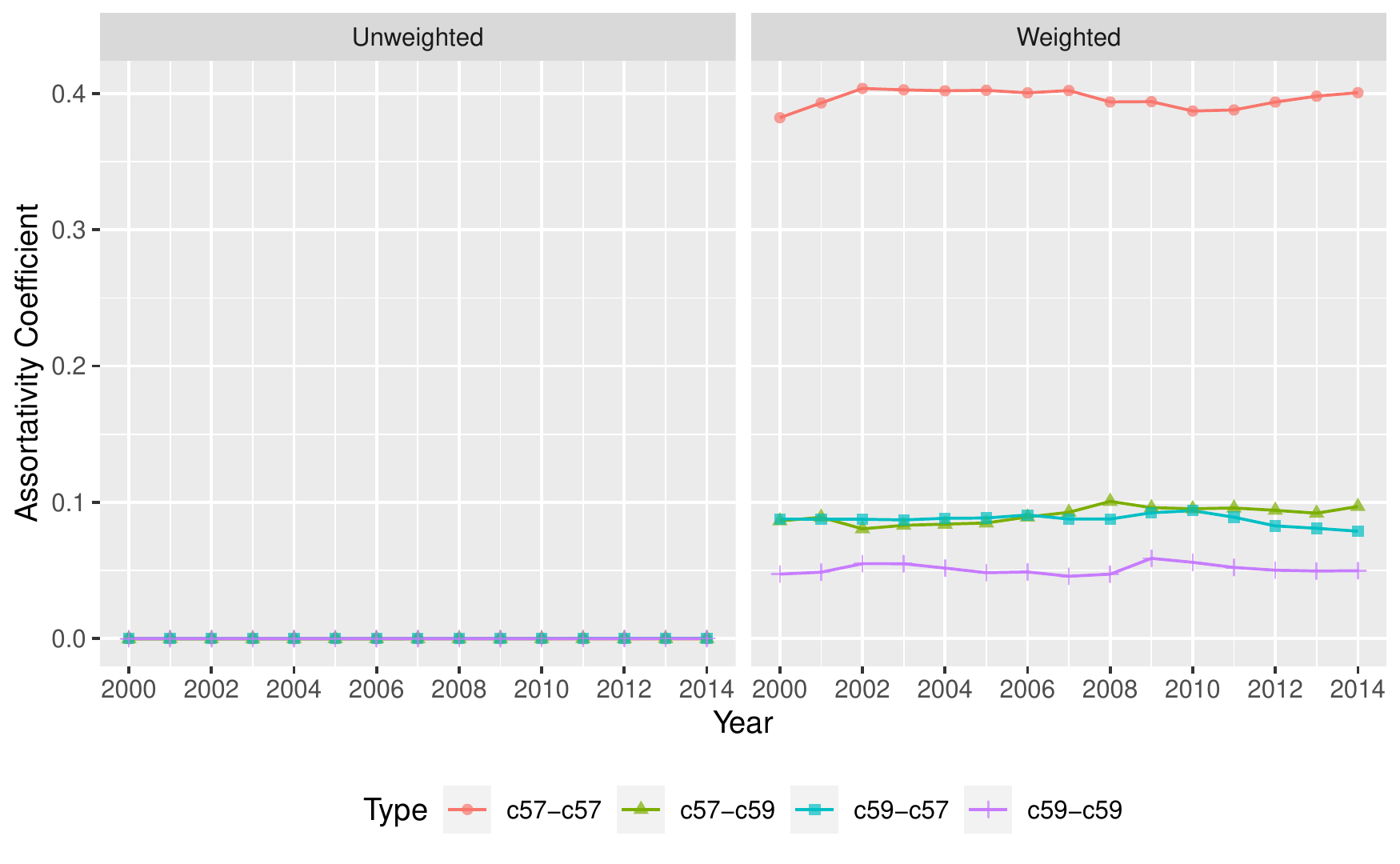}
	\caption{Assortativity coefficients between the final 
	consumption expenditure by households ($c57$ denoted by $X$) and 
	by government service ($c59$ denoted by $Y$) of WIONs from 2000 
	to 
	2014.}
	\label{fig:feature}
\end{figure}

The WIONs provide a good opportunity to illustrate the applicability of the
proposed assortativity measures to general vertex-specific features 
beyond
degree and strength. Consider two vertex features, the final consumption
expenditure by households ($c57$ denoted by $X$) and by government 
service ($c59$ denoted by $Y$).
These features are in the WIOTs but not used in constructing the WIONs.
Let $\rho_{\alpha, \beta}$, $\alpha, \beta\in \{X, Y\}$, be the assortativity
coefficient between feature $\alpha$ of the source vertex and 
feature $\beta$ of the
target vertex.  Figure~\ref{fig:feature}
shows the four assortativity coefficients (of both weighted and 
unweighted 
versions) for the periods of 2001--2014.
All the weighted assortativity coefficients are positive. The level of
$\rho_{X, X}$ is about 0.4, notably higher than the rest (0.1 or lower).
This implies that the region-sectors with higher final consumption by
households are likely transact more among themselves.
The curve of $\rho_{X, X}$  fluctuates within a relatively small region
over the 15 years. It has a noticeable increasing trend during 2000--2002,
followed by stable period before a downward trend during 2008--2010, and
rises steadily afterwards. The other three coefficients are closer to zero,
suggesting no strong tendency of assortative connections. Of interest is that
$\rho_{X, Y}$ and $\rho_{Y, X}$ stay together at about the same level over all
the years, despite that they measure transactions of opposite 
directions. If edge weight is not accounted when analyzing the 
correlation of the features independent of
network structure, all kinds of the assortativity measures are 
centered at $0$ with small variations (see the left panel in 
Figure~\ref{fig:feature}), and hence become 
non-informative and useless.

\section{Discussions}
\label{Sec:conc}

The proposed class of assortativity measures meet the practical need from
analyses of weighted and directed network. Special cases of the proposed
measures are equivalent to several classical assortativity measures
in the literature. The generalized measures can be
used to study the assortative tendency of any pair of
vertex-specific features.  Extensive simulations and two real
data analyses demonstrate that both edge direction and weight
should be accounted in the computation. Edge direction is an 
essential property, as a directed network may 
simultaneously present assortative and disassortative mixing, i.e., 
one kind of assortativity coefficient (e.g., out-in) may have a 
positive value, while another (e.g., in-out) has a negative value; 
see~\citet{Piraveenan2012assortative} for an example of food webs in 
the United States.
Edge weight characterizes weighted networks; discarded it leads to meaningless
results or incorrect conclusions.

The proposed assortativity measures, like Newman's original measure,
are still based on the concept of Pearson correlation coefficient.
Newman's measure converges to~$0$ when network size is sufficiently
large \citep{Litvak2013uncovering}, a drawback shared by the proposed
measures. In addition, our simulation examples in Section~\ref{Sec:bamodel}
suggest that the proposed measures
may not successfully characterize the assortativity of a network if
the edge weights all have similar levels. One possible remedy is to extend
assortativity measure based on Spearman's correlation
\citep{Litvak2013uncovering}. Along the same line, a few other
nonparametric correlation coefficients in statistics can similarly
be considered, such as Kendall's tau and Blomqvist's beta.
We will continue our research in this direction, and report our
outcomes elsewhere in the future.


\appendix

\section{An Illustrative Example}
\label{Append:example}

For example, let us consider a weighted and directed network as 
shown in~Figure~\ref{fig:toy}.
The weights are marked next to the directed edges; for example,
$w_{A B} = 10$, $\ins_A = 3$, $\outs_A = 13$, $\ins_B = 16$ and
$\outs_B = 9$. Additionally, we have
$\bar{s}_{\rm sou}^{\left(\rm out \right)} = 9.39$,
$\bar{s}_{\rm sou}^{\left(\rm in \right)} = 5.90$,
$\bar{s}_{\rm tar}^{\left(\rm out \right)} = 5.90$,
$\bar{s}_{\rm tar}^{\left(\rm in \right)} = 10.16$,
$\sigma_{\rm sou}^{\left(\rm out \right)} = 3.68$,
$\sigma_{\rm sou}^{\left(\rm in \right)} = 6.58$,
$\sigma_{\rm tar}^{\left(\rm out \right)} = 4.83$, and
$\sigma_{\rm tar}^{\left(\rm in \right)} = 6.06$.
For the given network, we have the weighted assortativity coefficients
$\rho_{\rm in, in} = -0.56$, $\rho_{\rm in, out} = -0.82$,
$\rho_{\rm out, in} = 0.29$ and $\rho_{\rm out, out} = -0.29$. This 
example network simultaneously presents assorative and 
disassortative 
mixing. However, when edge weights are not accounted in the 
computation, the corresponding four kinds of 
assortativity coefficients are all equal to $-0.75$, suggesting 
strong 
disassortative mixing.

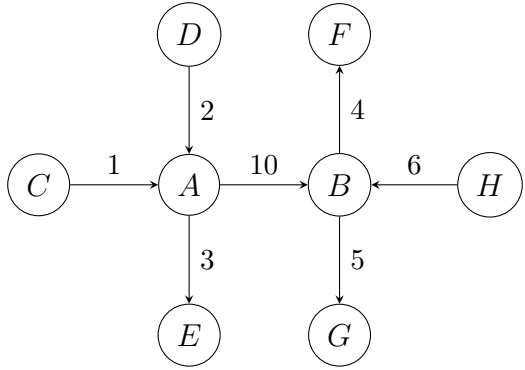
\begin{figure}[tbp]
  \centering
	\begin{tikzpicture}[auto,arrow/.style = {-stealth}]
	\matrix [matrix of math nodes,
	column sep={2cm,between origins},
	row sep={2cm,between origins},
	nodes={circle, draw, minimum size = 7.5mm}]
	{
		& |(D)| D & |(F)| F &        \\
		|(C)| C & |(A)| A & |(B)| B & |(H)| H\\
		& |(E)| E & |(G)| G &        \\
	};
	\begin{scope}[every node/.style = {font = \small\itshape}]
	\draw[arrow] (C) -- node [midway] {\rm 1} (A);
	\draw[arrow] (D) -- node [midway] {\rm 2} (A);
	\draw[arrow] (A) -- node [midway] {\rm 3} (E);
	\draw[arrow] (B) -- node [midway, right] {\rm 4} (F);
	\draw[arrow] (A) -- node [midway] {\rm 10} (B);
	\draw[arrow] (B) -- node [midway] {\rm 5} (G);
	\draw[arrow] (H) -- node [midway, above] {\rm 6} (B);
	\end{scope}
	\end{tikzpicture}
  \caption{An illustrative example of weighted directed network.}
  \label{fig:toy}
\end{figure}

\section{Variants of the Proposed Assortativity} 
\label{Append:math}

\begin{prop}
	\label{Prop:weightundirect}
	The proposed assortativity coefficient is equivalent to that
	given by \citet{Arcagni2019extending} (case 4) when a network is
	weighted but undirected.
\end{prop}
\begin{proof}
	For weighted but
	undirected networks, each weighted edge can be evenly split to
	two
	weighted and directed edges respectively pointing into the two
	vertices at the ends. In what follows, the strength of a node
	$i$,
	$s_i$, is equal to the sum of the manually created in-strength
	$\tildeins_i$ and out-strength $\tildeouts_i$, i.e.,
	$\tildeins_i
	= \tildeouts_i = s_i / 2$. The new arrangement causes a change
	in the weighted adjacency matrix; namely $\tilde{\wei} =
	\wei/2$ and consequently $\tilde{W} = \sum_{i,j \in V}
	\tilde{w}_{ij} = W/2$. As network direction is not accounted,
	the four kinds of assortativity measures are equivalent. Without
	loss of generality, we continue the verification by considering
	the ``out-out'' type. Elementary algebra leads us to
	\begin{equation*}
		\bar{s}^{\, \rm (out)}_{\rm sou} = \bar{s}^{\, \rm
			(out)}_{\rm tar} = \frac{\sum_{i,j \in V} (w_{ij} / 2)
			(s_i / 2)}{W/2} = \frac{\sum_{i,j \in V} w_{ij} s_i}{2W}
			= \frac{\mu}{2},
	\end{equation*}
	where $\mu$ is the weighted average of $s_i$. The standard
	deviation is given by
	\begin{equation*}
		\sigma^{\rm (out)}_{\rm sou} = \sigma^{\rm (out)}_{\rm tar}
		= \sqrt{\frac{\sum_{i,j} (w_{ij}/2)(s_i / 2 - \mu /
		2)^2}{W/2}} = \frac{1}{2} \cdot \sqrt{\frac{\sum_{i,j \in V}
		w_{ij} (s_i - \mu)^2}{W}} = \frac{\sigma}{2},
	\end{equation*}
	where $\sigma$ is the weighted standard deviation. In what
	follows, Equation~\eqref{eq:fourd} is reduced to
	\begin{equation*}
		\rho = \frac{\sum_{i,j \in V} (w_{ij}/2) \bigl[(s_i / 2 -
		\mu/2)(s_j/2 - \mu/2)\bigr]}{(W/2)(\sigma/2)(\sigma/2)} =
	\frac{\sum_{i,j \in V} w_{ij} \bigl[(s_i - \mu)(s_j -
	\mu)\bigr]}{W \sigma^2}
	\end{equation*},
which completes the proof.
\end{proof}

When a network $G$ is directed but unweighted, its adjacency matrix
$\wei = \adj$ is dichotomous, and the strength, in-strength and
out-strength of a vertex is equal to its degree, in-degree and
out-degree, respectively. Thus, Equation~\eqref{eq:fourd} is reduced
to
\begin{equation}
	\label{eq:fourud}
	\rho_{\alpha, \beta}(G) =
	\frac{
		\sum_{i, j \in V} a_{ij}
		\left[\left(d_i^{(\alpha)} - \bar{d}_{\rm sou}^{\, (\alpha)}
		\right)
		\left(d_j^{(\beta)} - \bar{d}_{\rm tar}^{\, (\beta)}
		\right)\right]
	}
	{A \sigma_{\rm sou}^{(\alpha)} \sigma_{\rm tar}^{(\beta)}},
\end{equation}
where $A:=\sum_{i,j \in V} a_{ij} = |E|$ counts the number of edges,
and the statistics $\bar{d}_{\rm sou}^{\, (\alpha)}$, $\bar{d}_{\rm
tar}^{\, (\alpha)}$, $\sigma_{\rm sou}^{(\alpha)}$ and $\sigma_{\rm
tar}^{(\beta)}$ are defined analogously. This is equivalent to
\citet[Equation~(1)]{Foster2010edge}. In addition,
Equation~\eqref{eq:fourud} is identical to
\citet[Equation~(25)]{Newman2003mixing} for
$(\alpha, \beta) = (\rm in, out)$, and respectively identical to
\citet[Equation~(21)]{Piraveenan2012assortative} and
\citet[Equation~(22)]{Piraveenan2012assortative} for
$(\alpha, \beta) = (\rm out, out)$ and
$(\alpha, \beta) = (\rm in, in)$ by using the
notations therein.

When a network $G$ is unweighted and undirected, we further simplify
Equation~\eqref{eq:fourd} to
\begin{equation*}
	\rho = \frac{\sum_{i,j \in V} a_{ij} \bigl[(d_i - \mu)(d_j -
	\mu)\bigr]}{A\sigma^2},
\end{equation*}
where $\mu = \sum_{i, j \in V} a_{ij} d_i/A$ and $\sigma =
\sqrt{\sum_{i,j \in V} a_{ij} (d_i - \mu)^2/(4A)}$. Elementary
algebra shows that the above expression is mathematically equivalent
to \citet[Equation~(4)]{Newman2002assortative}.

\section{Boundary Values of the Proposed Assortativity}
\label{Append:boundary}

Consider a network $G$ consisting of $n \ge 3$ vertices with
weighted adjacency matrix given by
\begin{equation*}
	\begin{pmatrix}
		0 & u_1 & 0 & 0 & \cdots & 0
		\\ 0 & 0 & u_2 & 0 & \cdots & 0
		\\ u_3 & 0 & 0 & 0 & \cdots & 0
		\\ \vdots & \vdots & \vdots & \vdots & \ddots & \vdots
		\\ 0 & 0 & 0 & 0 & \cdots & 0
	\end{pmatrix},
\end{equation*}
for some $u_1, u_2, u_3 > 0$. Elementary algebra shows that
$\rho_{{\rm in, out}}(G) = \rho_{{\rm out, in}}(G) = 1$, suggesting
that high out-strength vertices are all connected with high
in-strength ones.

Consider $G^{\ast}$ structurally similar to $G$ but with a different
weighted adjacency matrix given by
\begin{equation*}
	\begin{pmatrix}
		0 & u_1 & 0 & 0 & \cdots & 0
		\\ u_2 & 0 & u_3 & 0 & \cdots & 0
		\\ 0 & 0 & 0 & 0 & \cdots & 0
		\\ \vdots & \vdots & \vdots & \vdots & \ddots & \vdots
		\\ 0 & 0 & 0 & 0 & \cdots & 0
	\end{pmatrix}.
\end{equation*}
If $u_1$ is significantly smaller than $u_2$ and $u_3$ (like $u_1 =
o(\min\{u_2,u_3\})$ or $u_1 \to 0$), we have $\rho_{\rm out,
out}(G^{\ast}) \to -1$. For $u_3 = 0$ and $u_1 \neq u_2$, we
have $\rho_{\rm out, out}(G^{\ast}) = -1$.

\section{Rewiring Algorithm}
\label{Append:rewire}

Given the vertex strength distribution, i.e., $p_k = \Pr(S = z_k)$, 
$k = 1, \ldots, |\mathcal{S}|$, for $z_k \in \mathcal{S}$, if, in 
addition, we assume that the probability of an edge being sampled is 
proportional to its weight, then the conditional probability of 
corresponding contribution added to the strength of the vertex at 
either end of that edge is given by~$q_k = z_k p_k / 
\sum_{z_h \in 
\mathcal{S}}z_h p_h$. The algorithm of \citet{Newman2003mixing} 
requires a transition matrix
$\bm{M} := (m_{ij})$ of dimension $(|\mathcal{S}| \times |\mathcal{S}|)$
such that all the row sums and column sums are equal to~$0$ with
$\sum_{z_j, z_k   \in \mathcal{S}} z_j z_k m_{jk} = 1$.
The functionality of $\bm{M}$ is to bridge $q_k$ with a 
strength-based joint distribution for links. Specifically, let 
$\ell_{jk}$ be the probability that a vertex with strength~$z_j$
and a vertex with strength~$z_k$ is connected. We have
\begin{equation}
	\label{eq:linkdist}
	\ell_{jk} = q_j q_k + \xi \sigma_q^2 m_{jk},
\end{equation}
where $\sigma_q^2$ is the empirical variance of distribution $q$.
The collection of $\ell_{jk}$ forms a strength-based distribution of
edge connection, so $\ell_{jk}$ is required to fall in $[0, 1]$ for
all $j, k \in V$. Under the current setting, fundamental algebra can 
show that $\sum_{jk} l_{jk} = 1$ for $l_{jk}$ defined in 
Equation~\eqref{eq:linkdist}. 

Note that all the constraints for $m_{jk}$'s from $\bm{M}$ are 
linear. Though there may be more than one qualified $\bm{M}$ in 
practice, we shall choose an ``optimal'' $\bm{M}$ according to some 
sort of criterion such as matrix norm. In what follows, the 
construction of $\bm{M}$ has become a quadratic programming problem, 
which can be done through standard statistical programming software. 
Specifically, we used the \texttt{quadprog} package
\citep{Turlach2019quadprog} for \proglang{R}.

Now, we are ready to give the extended rewiring algorithm. The 
inputs for our algorithm are the pre-specified vertex strength 
distribution $p_k$ and the true assortativity $\xi$.
\begin{enumerate}
	\item Sample $n$ strength values for the $n$ vertices from the strength
          distribution $p_k$, $k = 1, \ldots, |\mathcal{S}|$.
	\item Generate a random network on the $n$ vertices as described 
	above.
	\item Randomly select a pair of edges respectively denoted by 
	$(v_a, v_b)$ and $(v_c, v_d)$, where $v_a$ and $v_b$, $v_c$ and 
	$v_d$ are the vertices at the two ends of the selected edges. 
	Correspondingly, let 
	$s_a$, $s_b$, $s_c$ and $s_d$ be the strengths of those 
	vertices. At each rewiring step, the probability that one unit 
	weight is added to $(v_a, v_c)$ and $(v_b, v_d)$, and 
	simultaneously one unit weight is taken away from $(v_a, v_b)$ 
	and 
	$(v_c, v_d)$ is given by
	\begin{equation*}
		\begin{cases}
			\displaystyle \frac{\ell_{s_{a} s_{c}} \ell_{s_{b} 
			s_{d}}}{\ell_{s_{a} s_{b}} \ell_{s_{c} s_{d}}}, \qquad 
			&\mbox{if }
				\ell_{s_{a} s_{c}} \ell_{s_{b} s_{d}} < 
				\ell_{s_{a} s_{b}} \ell_{s_{c} s_{d}};
			\\ 1, \qquad & \mbox{otherwise}.
		\end{cases}
	\end{equation*}
	Note that, if $(v_a, v_c)$ or $(v_b, v_d)$ does not exist, and 
	rewiring is executed, an edge (or two edges) is created; 
	Similarly, if $(v_a, v_b)$ or $(v_c, v_d)$ is a unit-weight 
	edge, and rewiring is executed, the edge (or both edges) is 
	removed.
	\item Repeat the rewiring procedure (in step 3) for a 
	preset number of times, and return the resultant 
	network. It is worth mentioning that the rewiring number needs 
	to be sufficiently large in general, where the present study 
	follows the guidance given by~\cite{Bertotti2019configuration}.
\end{enumerate}

\section{Backbone Extraction}
\label{Append:backbone}

The basic goal of backbone extraction is to filter out the
non-essential edges.  \citet{Xu2019input} applied
a disparity filter method \citep{Serrano2009extracting} to extract
backbones from input-output networks.
Consider the null hypothesis that the normalized weights
corresponding to the outgoing or incoming edges of a vertex are
the spacings formed by uniformly distributed variables over the unit interval.
Let $\normoutw_{ij} = w_{ij} / \outs_i$ and
$\norminw_{ji} = w_{ji} / \ins_i$
be the normalized out- and in-strength of vertex~$i$, respectively.
The number of outgoing and incoming edges of~$i$
are its out-degree~$\outd_i$ and in-degree~$\ind_i$, respectively.
Given that there are $d > 1$ out- or in-edges for a vertex, under the null
hypothesis, a normalized weight has density
\[
  p(x; d) = (d - 1)(1 - x)^{d - 2}, \quad 0 < x < 1, \quad d > 1.
\]
The p-values of edge~$i$ with normalized out-strength $\outd_i$
and in-strength $\ind_i$ are, respectively,
\begin{align*}
  \alpha_{ij}^{\, \rm (out)} = \int_{\normoutw_{ij}}^1
  p(x; \outd_i) \dd x
  \quad
  \mbox{ and }
  \quad 
  \alpha_{ij}^{\, \rm (in)} = \int_{\norminw_{ij}}^1 p(x;  \ind_i) \dd x.
\end{align*}
Borrowing the idea of hypothesis testing, we preserve edge $e_{ij}$ if
$\alpha_{ij}^{\, \rm (out)} < \alpha$ or $\alpha_{ij}^{\, \rm (in)} < \alpha$ as the
null assumption of uniform assignment is rejected; otherwise,
$e_{ij}$ is filtered out. In other words, an edge is preserved if
the criterion is satisfied for at least one of the two vertices at
the ends. For the very rare cases where $\outd_i = \ind_j = 1$, 
some special treatments based on additional criteria are needed, as 
removing these edges may break the connectivity of the network, but 
these edges themselves do not actually contribute heterogeneity to 
the network. These rare cases do not occur in the WIONs, so 
the standard filtering procedure suffices.

\end{document}